\documentclass{article}

\usepackage{amsmath,amssymb}

\usepackage{amsthm}

\newtheorem{thm}{Theorem}

\newcommand{\ra}{\rightarrow}
\def\To{\Rightarrow}
\def\TTo{\Longrightarrow}

\def\ie{i.e.}

\def\Qn{Q_{\mathcal N}}

\newcommand{\prule}[3]{\ensuremath{(#1,\,#2,\,#3)}}
\newcommand{\ins}[3]{\ensuremath{\prule{#1}{#2}{#3}_{ins}}}
\newcommand{\del}[3]{\ensuremath{\prule{#1}{#2}{#3}_{del}}}

\newcommand{\fins}[1]{\ins{\lambda}{#1}{\lambda}}
\newcommand{\fdel}[1]{\del{\lambda}{#1}{\lambda}}

\newcommand{\lcins}[2]{\ins{#1}{#2}{\lambda}}
\newcommand{\lcdel}[2]{\del{#1}{#2}{\lambda}}

\newcommand{\RCrule}[3]{\ensuremath{\bigl(#1, #2, #3\bigr)}}

\newcommand{\derivesby}[1]{\overset{#1}{\Longrightarrow}}
\newcommand{\dollarfour}{\$^{\scriptscriptstyle(IV)}}
\newcommand{\dollari}[1]{\$^{(#1)}}

\title{Random Context and Semi-Conditional Insertion-Deletion Systems}
\author{Sergiu Ivanov$^1$\and Sergey Verlan$^{1,2}$}

\date{
$^1$ Institute of Mathematics and Computer Science,\\
Academy of Sciences of Moldova,\\
Academiei 5, Chisinau, MD-2028, Moldova\\
email: sivanov@math.md\\
$^2$ Laboratoire d'Algorithmique, Complexit\'e et Logique,\\
Universit\'e Paris Est -- Cr\'eteil Val de Marne,\\
61, av. g\'en. de Gaulle, 94010 Cr\'eteil, France\\
email: verlan@univ-paris12.fr
}

\begin{document}

\maketitle

\begin{abstract}

In this article we introduce the operations of insertion and deletion working
in a random-context and semi-conditional manner. We show that the conditional
use of rules strictly increase the computational power. In the case of
semi-conditional insertion-deletion systems context-free insertion and deletion
rules of one symbol are sufficient to get the computational completeness. In
the random context case our results expose an asymmetry between the
computational power of insertion and deletion rules: systems of size $(2,0,0;
1,1,0)$ are computationally complete, while systems of size $(1,1,0;2,0,0)$
(and more generally of size $(1,1,0;p,1,1)$) are not. This is particularly
interesting because other control mechanisms like graph-control or matrix
control used together with insertion-deletion systems do not present such
asymmetry.
\end{abstract}

\section{Introduction}
Insertion-deletion systems are a powerful theoretical computational device
which is based on two elementary operations of insertion and deletion of
substrings in a string. These operations were first considered with a
linguistic motivation in~\cite{Marcus} and latter developed
in~\cite{Galiuk,Kluwer}. These references investigate Marcus contextual
grammars which capture many interesting linguistic properties like ambiguity
and duplication.

Another motivation for these operations can be found
in~\cite{Haussler82,Haussler83} where the insertion operation and its iterated
variant is introduced  as a generalization of Kleene's operations of
concatenation and closure~\cite{Kleene56}. The operation of concatenation would
produce a string $xyz$ from two strings $xy$ and $z$. By allowing the
concatenation to happen anywhere in the string and not only at its right
extremity a string $xzy$ can be produced, i.e., $z$ is inserted into $xy$.
In~\cite{Kari} the deletion is defined as a right quotient operation which
happens not necessarily at the rightmost end of the string.

The third inspiration for insertion and deletion operations comes,
surprisingly, from the field of molecular biology. In fact they correspond to a
mismatched annealing of DNA sequences, see~\cite{dna} for more details.  Such
operations are also present in the evolution processes under the form of point
mutations as well as in RNA editing, see the discussions
in~\cite{Beene,BBD07,Smith,dna}. This biological motivation of
insertion-deletion operations led to their study in the framework of molecular
computing, see, for example, \cite{Daley,cross,dna,TY}.

In general, an insertion operation means adding a substring to a given string
in a specified (left and right) context, while a deletion operation means
removing a substring of a given string from a specified (left and right)
context. A finite set of insertion-deletion rules, together with a set of
axioms provide a language generating device: starting from the set of initial
strings and iterating insertion-deletion operations as defined by the given
rules, one obtains a language.

Even in their basic variants, insertion-deletion systems are able to
characterize the recursively enumerable languages, see~\cite{VerlanH} for an
overview of known results. Moreover, as it was shown in \cite{cfinsdel}, the
context dependency may be replaced by insertion and deletion of strings of
sufficient length, in a context-free manner. If the length is not sufficient
(less or equal to two) then such systems are not able to generate more than the
context-free languages and a characterization of them was shown
in~\cite{SV2-2}.

Similar investigations were continued in \cite{MRV07,KRV08,KRV08c} on
insertion-deletion systems with one-sided contexts, i.e., where the context
dependency is asymmetric and is present only from the left or only from the
right side of all insertion and deletion rules. The papers cited above give
several computational completeness results depending on the size of insertion
and deletion rules. We recall the interesting fact that some combinations are
not leading to computational completeness, i.e., there are
languages that cannot be generated by such devices.

Similarly as in the case of context-free rewriting, it is possible to consider
a graph-controlled variant of insertion-deletion systems. Thus the rules cannot
be applied at any time, as their applicability depends on the current
``state'', changed by a rule application. Such a formalization is rather
similar to the definition of insertion-deletion P systems~\cite{membr}.
One-sided graph-controlled insertion-deletion systems where at most two symbols
may be present in the description of insertion and deletion rules were
investigated in~\cite{FKRV10}. This correspond to systems of size
$(1,1,0;1,1,0)$, $(1,1,0;1,0,1)$, $(1,1,0;2,0,0)$, and $(2,0,0;1,1,0)$, where
the first three numbers represent the maximal size of the inserted string and
the maximal size of the left and right contexts, resp., while the last three
numbers represent the same information for deletion rules. It is known that
such systems are not computationally complete~\cite{KRV10}, while the
corresponding P systems and graph-controlled variants are. A summary of these
results can be found in\,\cite{AKRV10b,VerlanH}.

Further adaptation of control mechanisms from the area of regulated rewriting
leads us to matrix insertion-deletion systems~\cite{PV10}, where insertion and
deletion rules are grouped in sequences, called matrices, and either the whole
sequence is applied consecutively, or no rule is applied. It was shown that in
the case of such control the computational power of systems of sizes above is
strictly increasing. Moreover, binary matrices suffice to achieve this result
giving a  characterization similar to the binary normal form for matrix
grammars.

In this paper we continue the investigation on the ``regulated''
insertion-deletion. We adapt the idea of semi-contextual and random context
grammars to these systems. More precisely, a permitting and forbidding
condition is associated to each insertion and deletion rule. In the case of
semi-contextual control the conditions are based on sets of words, while in the
random context case the conditions are sets of letters. A rule can be applied
if all the words from the permitting condition are present in the string, while
no word from the forbidding condition is present. We show that the
computational completeness can be achieved by semi-conditional
insertion-deletion systems whose permitting and forbidding conditions contain
words of length at most~2 and using context-free insertion and deletion of one
symbol, \ie{} rules of size $(1,0,0;1,0,0)$. This is quite interesting result
showing that the semi-conditional control is very powerful -- even in the case
of graph-controlled insertion-deletion systems with appearance check only
recursively enumerable sets of numbers can be obtained with same insertion and
deletion parameters.

Until now all the results in the area of insertion-deletion systems (with or
without control mechanisms) exhibited the property that the computational
completeness is achieved or not at the same time by both systems of size
$(n,m,m';p,q,q')$ and $(p,q,q';n,m,m')$, \ie{} when the parameters of insertion
and deletion are interchanged. This empirical property suggested that with
respect to the computational completeness the insertion and deletion rules are
similar (even if the proofs use quite different ideas). This article shows for
the first time that this property does not always hold, the counterexample
being random context insertion-deletion systems of size $(2,0,0;1,1,0)$ which
generate all recursively enumerable languages, while systems of size
$(1,1,0;2,0,0)$ (and more generally of size $(1,1,0;p,1,1)$) cannot do this.

\section{Preliminaries}

We do not present the usual definitions concerning standard concepts of the
theory of formal languages and we only refer to \cite{handbook} for more
details.
The empty string is denoted by $\lambda $. The family of recursively
enumerable, context-sensitive and context-free languages is denoted by $RE$,
$CS$ and $CF$, respectively. We will denote the length of a string $w$ with
$|w|$; we will refer to the number of occurrences of the symbol $a$ in the
string $|w|$ by the notation $|w|_a$.

In the following, we will use special variants of the \emph{Geffert} normal
form for type-0 grammars (see~\cite{Geffert91} for more details).

A grammar $G=\left( N,T,S,P\right) $ is said to be in \emph{Geffert normal
form}~\cite{Geffert91} if the set of non-terminal symbols $N$ is defined as
$N=\{S,A,B,C,D\}$, $T$ is an alphabet and $P$ only contains context-free rules
of the forms $S\to uSv$ with $u\in \{A,C\}^+$ and $v\in (T\cup \{B,D\})^+$ as
well as $S\to \lambda$
and two (non-context-free) erasing rules $AB\to \lambda $ and $CD\to \lambda $.

We remark that in the Geffert normal form we can easily transform the linear
rules (suppose that $A$, $B$, $C$ and $D$ are treated like terminals) into a
set of left-linear and right-linear rules (albeit by increasing the number of
non-terminal symbols, e.g., see \cite{membr}). More precisely, we say that a
grammar $G=\left( N,T,S,P\right) $ with $N=N'\cup N''$, $S,S'\in N'$, and
$N''=\{A,B,C,D\}$, is in the \emph{special Geffert normal form} if, besides the
two erasing rules $AB\to \lambda $ and $CD\to \lambda $, it only has
context-free rules of the following forms:
\begin{align*}
& X\to bY,\quad \text{where }X,Y\in N',b\in N'', X\ne Y\\
& X\to Yb,\quad \text{where }X,Y\in N',b\in T\cup N'', X\ne Y\\
& S'\to \lambda .
\end{align*}

Moreover, except rules of the forms  $X\to Sb$ and $X\to S'b$, we may even
assume that for the first two types of rules it holds that the right-hand side
is unique, i.e., for any two rules $X\to w$ and $U\to w$ in $P$ we have $U=X$.

The computation in a grammar in the special Geffert normal form is done in two
stages. During the first stage, only context-free rules are applied. During the
second stage, only the erasing rules $AB\to \lambda $ and $CD\to \lambda $ are
applied. These two erasing rules are not applicable during the first stage as
long as the left and the right part of the current string are still separated
by $S$ (or $S'$) as all the symbols $A$ and $C$ are generated on the left side
of these middle symbols and the corresponding symbols $B$ and $D$ are generated
on the right side. The transition between stages is done by the rule $S'\to
\lambda $ (which corresponds to the rule $S\to\lambda$ from the Geffert normal
form).
We remark that all these features of a grammar in the special Geffert normal
form are immediate consequences of the proofs given in~\cite{Geffert91}.

Since during the first stage symbols $A$, $B$, $C$ and $D$ act like
non-terminals, if no confusion arises, we consider that during this stage there
is only one non-terminal and there are only left- and right-linear rules.

We give below the definitions of semi-conditional and random-context grammars,
as given in~\cite{handbook}.

A {\it semi-conditional} grammar is a quadruple $G=(N,T,S,P)$, where $N,T,S$
are specified as in a context-free grammar  and $P$ is a finite set of triples
of the form $p=( A\ra w;E,F),$ where $A\ra w$ is a context-free production over
$N\cup T$ and $E,F$ are finite subsets of $ (N\cup T)^+$. Then, $p$ can be
applied to a string $x\in (N\cup T)^*$ only if $A$ appears in $x$, each element
of $E$ and no element of $F$ is a subword of $x$. If $E$ or $F$ is the empty
set, then no condition is imposed by $E,$ or respectively, $F$. $E$ is said to
be the set of {\it permitting} and $F$ is said to be the set of {\it
forbidding} context conditions of $p.$

We remark that in some sources, e.g.~\cite{paun97}, the definition of
semi-conditional grammars implies that $|E|\le 1$, $|F|\le 1$. We consider the
definition from~\cite{handbook} where the above restriction is not present. We
denote by $SC_{i,j}$ the family of semi-conditional grammars where $|w_1|\le
i$, $|w_2|\le j$, for any $w_1\in E$ and $w_2\in F$ for any rule $(A\to w;E,F)$
of the grammar.

 A \emph{random context} grammar is a semi-conditional
grammar where $E,F\subseteq N$.

Semi-conditional grammars generate all recursively enumerable or all
context-sensitive languages, depending on whether $\lambda$-rules are used or
not, respectively.

\subsection{Insertion-deletion systems}

An \textit{insertion-deletion system} is a construct $ID=(V,T,A,I,D)$, where
$V$ is an alphabet; $T\subseteq V$ is the \textit{terminal} alphabet (the
symbols from $V\setminus T$ are called \textit{non-terminal} symbols);
$A\subseteq V^*$ is the set of \textit{axioms}; $%
I,D$ are finite sets of triples of the form $(u,\alpha ,v)$, where $u$, $%
\alpha $ ($\alpha \neq \lambda $), and $v$ are strings over $V$. The triples in
$I$ are \textit{insertion rules}, and those in $D$ are \textit{deletion
rules}. An insertion rule $(u,\alpha ,v)\in I$ indicates that the string $%
\alpha $ can be inserted between $u$ and $v$, while a deletion rule $%
(u,\alpha ,v)\in D$ indicates that $\alpha $ can be removed from between the
context $u$ and $v$. Stated in another way, $(u,\alpha ,v)\in I$ corresponds to
the rewriting rule $uv\to u\alpha v$, and $(u,\alpha ,v)\in D$ corresponds to
the rewriting rule $u\alpha v\to uv$. By $\To _{ins}$ we denote the relation
defined by the insertion
rules (formally, $x\To_{ins}y$ if and only if $%
x=x_{1}uvx_{2},y=x_{1}u\alpha vx_{2}$, for some $(u,\alpha ,v)\in I$ and
$x_{1},x_{2}\in V^*$), and by $\To_{del}$ the relation defined by the deletion
rules (formally, $x\To_{del}y$ if and only if $x=x_{1}u\alpha
vx_{2},y=x_{1}uvx_{2}$, for some $(u,\alpha ,v)\in D$ and $x_{1},x_{2}\in
V^*$). By $\To $ we refer to any of the relations $\To_{ins},\To_{del}$, and by
$\To^*$ we denote the reflexive and transitive closure of $\To$.

We will often consider $R = I\cup D$ and distinguish between insertion and
deletion rules by the subscripts $_{ins}$ or $_{del}$.

The language generated by $ID=(V,T,A,I,D)$ is defined by
\begin{equation*}
L(ID)=\{w\in T^*\mid x\To ^*w\mathrm{\ for\ some\ }
x\in A\}.
\end{equation*}

The complexity of an insertion-deletion system $ID=(V,T,A,I,D)$ is described by
the  vector\\ $(n,m,m';p,q,q')$ called \emph{size}, where
\begin{eqnarray*}
n=\max\{|\alpha|\mid (u,\alpha,v)\in I\}, & & p=\max\{|\alpha|\mid
(u,\alpha,v)\in D\}, \\
m=\max\{|u|\mid (u,\alpha,v)\in I\}, & & q=\max\{|u|\mid (u,\alpha,v)\in D\},
\\
m'=\max\{|v|\mid (u,\alpha,v)\in I\}, & & q'=\max\{|v|\mid
(u,\alpha,v)\in D\}.
\end{eqnarray*}

The \emph{total size} of an insertion-deletion system $ID$ of size
$(n,m,m';p,q,q')$ is defined as the sum of all the numbers from the vector:
$\Sigma(ID)=n+m+m'+p+q+q'$.

By $INS_{n}^{m,m'}DEL_{p}^{q,q'}$ we denote the families of languages generated
by insertion-deletion systems having the size $(n,m,m';p,q,q')$.

If one of the parameters $n,m,m',p,q,q'$ is not specified, then instead we
write the symbol~$\ast $. In particular, $INS_*^{0,0}DEL_*^{0,0}$ denotes the
family of languages generated by \emph{context-free insertion-deletion
systems}. If one of the numbers from the pairs $m $, $m'$ and/or $q$, $q'$ is
equal to zero (while the other one is not), then we say that the corresponding
families have a one-sided context. Finally we remark that the rules from $I$
and $D$ can be put together into one set of rules $R$ by writing $\left(
u,\alpha ,v\right) _{ins}$ for $\left( u,\alpha ,v\right) \in I$ and $\left(
u,\alpha ,v\right) _{del}$ for $\left( u,\alpha ,v\right) \in D$.

\subsection{Conditional insertion-deletion systems}

In a similar way to context-free grammars, insertion-deletion systems can be
extended by adding some additional controls.
We discuss here the adaptation of the idea of semi-conditional and random
context grammars for insertion-deletion systems  and define corresponding types
of insertion-deletion systems.

A \emph{semi-conditional insertion-deletion system} of degree $(i,j), i,j\ge1$ is a construct%
\begin{equation*}
\gamma =(V,T,A,R),\mathrm{\ where}
\end{equation*}
\begin{itemize}
\item $V$ is a finite alphabet,

\item $T\subseteq V$ is the \emph{terminal alphabet},

\item $A\subseteq V^*$ is a finite set of \emph{axioms},

\item $R=\{r_1,\dots, r_n\}$, $n\ge 1$ is a finite set of rules of the form
    $r_k:\RCrule{\prule{a}{b}{c}_t}{P}{Q}$, where $\prule{a}{b}{c}$, is an
    insertion rule if $t=ins$ or a deletion rule if $t=del$, $P\subseteq
    V^i$, $Q\subseteq V^j$, $1\le k\le n$.
\end{itemize}

The sentential form (also called configuration) of $\gamma$ is a string $w\in
V^*$. For $r_k:\RCrule{\prule{a}{b}{c}_t}{P}{Q}$ a transition $w\TTo_{r_k}w'$,
for $1\le k\le n$, is performed if $w\To_t w'$ ($t\in\{ins,del\}$) and for all
$x\in P$, $x$ is a subword of $w$ and for no $y\in Q$, $y$ is a subword of $w$.

The language generated by $\gamma$ is defined by
\begin{equation*}
L(\gamma)=\{w\in T^*\mid x\TTo ^*w\mathrm{\ for\ some\ }
x\in A\}.
\end{equation*}

For $i,j>1$ we denote by $SC_{i,j}INS_{n}^{m,m'}DEL_{p}^{q,q'}$  the families
of languages generated by semi-conditional insertion-deletion systems of degree
$(i,j)$  having insertion and deletion rules of size\linebreak
$(n,m,m';p,q,q')$.

We also define \emph{random context} insertion-deletion systems as
semi-conditional insertion-deletion systems where $i=j=1$. They are denoted with
$RC\,INS_{n}^{m,m'}DEL_{p}^{q,q'}$.

\section{Semi-conditional insertion-deletion systems}

In this section we show that semi-conditional insertion-deletion systems of
degree $(2,2)$ are computationally complete even if rules inserting or deleting
a single symbol are used. We start by showing that that this result is optimal
with respect to the size of insertion and deletion rules.

\begin{thm}\label{thm:sc10}
$SC_{2,2}INS_1^{0,0}DEL_0^{0,0}\subsetneq RE$.
\end{thm}

\begin{proof}
Let $\gamma=(T,T,A,R)$, $\gamma\in INS_1^{0,0}DEL_0^{0,0}$ be an insertion-deletion
system. Consider the following context-free grammar $G=(\{S\},T,S,P)$,
$P=\{S\to Sa_1S\dots{}Sa_n\mid a_1\dots{}a_n\in A\}\cup \{S\to SaS\mid a\in
T\}\cup \{S\to\lambda\}$. Let $G'$ be the grammar obtained from $G$ by the
elimination of $\lambda$-productions. Since $\lambda\not\in L(\gamma)$, $G'$
does not contain any erasing rules. This combined with the fact that the family
of semi-contextual grammars without $\lambda$-productions is included in $CS$
yields the result.
\end{proof}

\begin{thm}\label{thm:sc11}
$SC_{2,2}INS^{0,0}_1DEL^{0,0}_1=RE$.
\end{thm}

\begin{proof}

We will start by proving the inclusion $SC_{2,2}INS^{0,0}_1DEL^{0,0}_1\supseteq
RE$. To achieve this goal, we will show the inclusion
$SC_{2,2}INS^{0,0}_1DEL^{0,0}_1\supseteq \lambda RC_{ac}$, where $\lambda
RC_{ac}$ is the family of random context grammars.

Consider an arbitrary random context grammar $G=(V,T,S,R)$. The set $R$
contains context-free rules with permitting and forbidding contexts.

We can assume that for any random context rule $r:\RCrule{A\ra u}{P}{Q}$ of
$G$, either $|u|=2$, or $|u|=0$. Indeed, any rule $r:\RCrule{A\ra u_1u_2\ldots
u_n}{P}{Q}$ can be replaced by a set of following random context rules:
\[
\begin{array}{l}
\{\RCrule{A\ra u_1 W^{(1)}_r}{P}{Q\cup
  Q_W}\}\cup\{\RCrule{W^{(i)}_r\ra u_{i+1}
  W^{(i+1)}_r}{\emptyset}{\emptyset}|1\leq i\leq n -
1\}\cup\\ \cup\{\RCrule{W^{(n)}_r\ra \lambda}{\emptyset}{\emptyset}\}.
\end{array}
\]
The set $Q_W$ above is defined as follows:
\[
Q_W = \{W^{(j)}_r\mid \exists r:\RCrule{B\ra\beta}{P'}{Q'}\in R, 1\leq j\leq n_r\},
\]
where $n_r$ is the number of symbols in the right-hand side of the rule with
label $r$.

It can be easily seen that the rules above will simulate the context-free rule
$A\to u_1\dots u_n$. The context checking is performed  during the application
of the rule $A\ra u_1 W_r^{(1)}$. The presence of $Q_W$ in the forbidding set
ensures that once started, the sequence of related rules shall be terminated.



Before proceeding to the simulation we will do some preliminary considerations.
First of all, the central problem with $SC_{2,2}INS^{0,0}_1DEL^{0,0}_1$ systems
is that there is no direct way to check the context of a specific symbol in the
string.  To address this problem, we encode every symbol $a\in V$ with a pair
of symbols $\hat{a}\bar{a}$ and place special terminators $B$ and $E$ at the
beginning and the end of the string correspondingly.  Every rule will have a
special forbidding context which will check whether the string has this form.
We will refer to this forbidding context as to the {\em normalization
condition}.

Now, to operate at a specific locus in the string, we insert some service
symbol and use the permitting/forbidding contexts of the subsequent rules to
check whether it is located in the appropriate position.  We then insert and
delete symbols as we need, and the normalization condition included in every
rule will assure that the operations we are performing are only happening in
some neighborhood of the service symbols: whenever the proper organization of
the string is broken, no rules will be further applicable, thus blocking the
computation at a string which contains non-terminal symbols.

Having done the overview of our approach, we are ready to present the solution
itself.

Let $c:V\to (\hat{V}\cup\bar{V})^2$ be a coding defined as
$c(a)=\hat{a}\bar{a}$, $a\in V$.

The normalization condition is the following forbidding condition added to the
forbidding contexts of the majority of rules:
\[
\Qn=\{\hat{x}\hat{y}, \bar{x}\bar{y}\mid x, y\in V\}\cup
\{B\bar{x}, \hat{x}E\mid x\in V\}\cup
\{u B, E u\mid u\in \hat{V}\cup\bar{V}\}.
\]

It can be seen that if $V_1=\hat{V}\cup\bar{V}\cup\{B,E\}$ then
$V_1^*\setminus \Qn=\{B(\hat{x}\bar{y})^+E\mid x,y\in V\}$. Indeed, the first
group of restrictions requires that the string only contains an alternating
sequence of hatted and barred symbols, the third group requires the string to
begin by $B$ and to end by $E$, while the second group states that the first
symbol after $B$ has a hat and the symbol before $E$ has a bar. A string of the
form above is said to be in a normalized form.

We construct the following semi-conditional insertion-deletion system
$\Gamma=(V_\Gamma,T_\Gamma,A_\Gamma,R_\Gamma)$.

The terminal alphabet $T_\Gamma = \{\hat{a}\mid a\in V\}$ contains the hatted
versions of original terminals. The alphabet $V_\Gamma$ includes contains the
service symbols, the hatted and barred versions of every symbol in the alphabet
of $G$, and the new terminal symbols:
\[
\begin{array}{rcl}
Q_\#&=&\{\#_p\mid \exists p:\RCrule{A\ra\lambda}{P}{Q}\in R\},\\
Q_\$&=&\{\dollari{i}_q\mid \exists q:\RCrule{A\ra\alpha}{P}{Q}\in R, \alpha\neq\lambda, 1\leq i\leq 5\},\\
V_\Gamma&=&Q_\#\cup Q_\$\cup \{\hat{a}, \bar{a}\mid a\in V\}\cup T_\Gamma.
\end{array}
\]
The only axiom of the insertion-deletion system is
$A_\Gamma=\{B\hat{S}\bar{S}E\}$.

Note that we are defining the new terminal alphabet $T_\Gamma$ instead of
re-using the old $T$ only for the sake of keeping the description simple and
that this change does not affect the result itself.

The set of rules $R_\Gamma$ is constructed in the following way:
\begin{itemize}
\item for every rule $p:\RCrule{X\ra\lambda}{P}{Q}\in R$ we add to
    $R_\Gamma$ the rules
  \begin{align*}
    &p.1:\RCrule{\fins{\#_p}}{\{\hat{X}\bar{X}\}\cup P}{Q_\#\cup Q_\$\cup Q\cup \Qn}\\
    &p.2:\RCrule{\fdel{\hat{X}}}{\{\#_p\hat{X}\}}{\Qn}\\
    &p.3:\RCrule{\fdel{\bar{X}}}{\{\#_p\bar{X}\}}{\{\hat{X}\#_p\}\cup \Qn}\\
    &p.4:\RCrule{\fdel{\#_p}}{\emptyset}{\Qn},
  \end{align*}

\item for every rule $q:\RCrule{X\ra Y Z}{P}{Q}\in R$, $X, Y, Z\in V$ we
    add to $R_\Gamma$ the rules
  \begin{align*}
    q.1:&\RCrule{\fins{\dollari{1}_q}}{\{\hat{X}\bar{X}\}\cup P}{Q_\#\cup Q_\$\cup Q \cup \Qn},\\
    q.2:&\RCrule{\fins{\dollari{2}_q}}{\{\dollari{1}_q \hat{X}\}}{\{\dollari{2}_q\}\cup \Qn},\\
    q.3:&\RCrule{\fdel{\hat{X}}}{\{\dollari{1}_q \hat{X}, \bar{X}\dollari{2}_q\}}{\Qn},\\
    q.4:&\RCrule{\fins{\dollari{3}_q}}{\{\dollari{1}_q \bar{X}\}}{\{\dollari{2}_q \bar{X},\dollari{3}_q,\hat{X} \dollari{1}_q\}\cup \Qn},\\
    q.5:&\RCrule{\fdel{\dollari{1}_q}}{\{\dollari{3}_q \dollari{1}_q\}}{\Qn},\\
    q.6:&\RCrule{\fdel{\bar{X}}}{\{\dollari{3}_q \bar{X}\}}{\{\dollari{1}_q\}\cup \Qn},\\
    q.7:&\RCrule{\fins{\dollari{4}_q}}{\{\dollari{3}_q\dollari{2}_q\}}{\{\dollari{1}_q,\dollari{4}_q\}\cup \Qn},\\
    q.8:&\RCrule{\fdel{\dollari{2}_q}}{\{\dollari{4}_q\dollari{3}_q\}}{\Qn},\\
    q.9:&\RCrule{\fins{\dollari{5}_q}}{\{\dollari{4}_q\}}{\{\dollari{2}_q,\dollari{5}_q,\hat{X}\dollari{4}_q\}\cup \Qn},\\
    q.10:&\RCrule{\fins{\hat{Y}}}{\{\dollari{4}_q\dollari{3}_q, \dollari{3}_q\dollari{5}_q\}}{\{\hat{Y}\dollari{4}_q\}\cup \Qn},\\
    q.11:&\RCrule{\fins{\bar{Y}}}{\{\dollari{4}_q\hat{Y}, \hat{Y}\dollari{3}_q, \dollari{3}_q\dollari{5}_q\}}{\{\dollari{5}_q\bar{Y}\}\cup \Qn},\\
    q.12:&\RCrule{\fins{\hat{Z}}}{\{\dollari{4}_q\hat{Y}, \bar{Y}\dollari{3}_q, \dollari{3}_q\dollari{5}_q\}}{\{\hat{Z}\dollari{4}_q\}\cup \Qn},\\
    q.13:&\RCrule{\fins{\bar{Z}}}{\{\dollari{4}_q\hat{Y}, \bar{Y}\dollari{3}_q, \dollari{3}_q\hat{Z}, \hat{Z}\dollari{5}_q\}}{\{\dollari{5}_q\bar{Z}\}\cup \Qn},\\
    q.14:&\RCrule{\fdel{\dollari{3}_q}}{\{\dollari{4}_q\hat{Y}, \bar{Y}\dollari{3}_q, \dollari{3}_q\hat{Z}, \bar{Z}\dollari{5}_q\}}{\Qn},\\
    q.15:&\RCrule{\fdel{\dollari{4}_q}}{\{\dollari{4}_q\hat{Y},\bar{Z}\dollari{5}_q\}}{\{\dollari{3}_q\}\cup \Qn},\\
    q.16:&\RCrule{\fdel{\dollari{5}_q}}{\emptyset}{\{\dollari{3}_q, \dollari{4}_q\}\cup \Qn},
  \end{align*}

\item we also add the rules
  \begin{align*}
    \{&\RCrule{\fdel{B}}{\emptyset}{V_\Gamma\backslash\left(\{B,E\}\cup\{\hat{a},\bar{a}\mid a\in T\}\right)},\\
    &\RCrule{\fdel{E}}{\emptyset}{\{B\}}\}\\
    \cup\{&\RCrule{\fdel{\bar{a}}}{\emptyset}{\{B, E\}}\mid a\in V\backslash T\}.
  \end{align*}
\end{itemize}

We will start our explanations with analyzing the third and simplest group of
rules. These rules permit to consecutively erase symbols $B$, $E$ and
$\bar{x}$, where $x\in T$. The first forbidding condition ensures that the group is
applicable only for strings of form $w=B(\hat{x}\bar{y})^*E$, $x,y\in T$. So
the result of the application of this group of rules on $w$ is a projection of
$w$ over $\hat{V}$.

We will now turn to simulating the erasing random context rule
$p:\RCrule{X\ra\lambda}{P}{Q}$.  The correct simulation sequence is as follows:
\begin{gather*}
  w_1 \hat{X}\bar{X} w_2\derivesby{p.1} w_1 \#_p \hat{X}\bar{X}
  w_2\derivesby{p.2} w_1 \#_p \bar{X} w_2\derivesby{p.3} w_1\#_p w_2
  \derivesby{p.4} w_1 w_2,
\end{gather*}
where $w_1, w_2\in V_\Gamma^*$.  Note that if $w_1 \hat{X}\bar{X} w_2$
satisfies the normalization condition, then $w_1w_2$ also satisfies it.

It turns out that almost any other derivation sequence will result in an
unproductive string with non-terminals.  Indeed, the permitting context of rule
$p.2$ will make it inapplicable if $\#_p$ has not been properly positioned by
$p.1$.  Further, if $p.2$ erases the wrong $\hat{X}$ (i.e., not the $\hat{X}$
located immediately to the right of $\#_p$), then the string will not satisfy
the normalization condition and the computation will halt at a string with
non-terminals.  The contexts of the rule $p.3$ assure that $\bar{X}$ is erased
only after the correct instance of $\hat{X}$ has been erased by $p.2$.  Note
that the forbidding word $\hat{X}\#_p$ prevents $p.3$ from being applicable if
$p.1$ inserts $\#_p$ between $\hat{X}$ and $\bar{X}$.  Finally, if $p.4$ erases
$\#_p$ when only $\hat{X}$ has been erased, the string will no longer satisfy
the normalization condition; on the other hand, if $\#_p$ is erased immediately
after it has been inserted, no effect is produced at all.

Now observe that, after the applications of $p.1$, $p.2$, and $p.3$ two valid
scenarios are possible.  If $w_2\neq \hat{X}\bar{X} w'_2$, $w'_2\in
V_\Gamma^*$, then the only applicable rule is $p.4$, because the $\#_p$ blocks
applications of any other rules.  If, however, $w_2 = \hat{X}\bar{X} w'_2$,
then the rule $p.2$ is applicable again, together with $p.4$.  Thus, instead of
erasing $\#_p$ by $p.4$, the system may erase one more pair $\hat{X}\bar{X}$.
This is not a problem, though, because this would correspond to the simulation
of another application of the rule $p$.

Note that it is not necessary to check for the permitting and forbidding
contexts of the original rule $p$ ($P$ and $Q$ correspondingly) before the
application of the rule $p.2$, even if $p.2$ may start as many simulations of
the erasing rule $p$ as there are consequent instances of $\hat{X}\bar{X}$ at
the locus marked by $\#_p$.  The reason is that $p.2$ is only applicable after
an application of $p.1$ (which adds a $\#_p$), and $p.1$ includes both $P$ and
$Q$ in its contexts already.

Quite as expected, the simulation of the generic context-free rule $q:(X\ra
YZ,P,Q)$ is the most complicated.  As usual, we will start with showing the
correct simulation sequence:
\begin{multline*}
w_1 \hat{X}\bar{X} w_2 \derivesby{q.1} w_1 \dollari{1}_q\hat{X}\bar{X} w_2
\derivesby{q.2} w_1 \dollari{1}_q\hat{X}\bar{X} \dollari{2}_q w_2
\derivesby{q.3} w_1 \dollari{1}_q\bar{X} \dollari{2}_q w_2 \\ \derivesby{q.4}
w_1 \dollari{3}_q\dollari{1}_q\bar{X} \dollari{2}_q w_2 \derivesby{q.5} w_1
\dollari{3}_q\bar{X} \dollari{2}_q w_2 \derivesby{q.6} w_1 \dollari{3}_q
\dollari{2}_q w_2 \\\derivesby{q.7} w_1 \dollari{4}_q \dollari{3}_q
\dollari{2}_q w_2 \derivesby{q.8} w_1 \dollari{4}_q \dollari{3}_q w_2
\derivesby{q.9} w_1 \dollari{4}_q \dollari{3}_q \dollari{5}_q w_2
\\\derivesby{q.10} w_1 \dollari{4}_q \hat{Y} \dollari{3}_q \dollari{5}_q w_2
\derivesby{q.11} w_1 \dollari{4}_q \hat{Y}\bar{Y} \dollari{3}_q \dollari{5}_q
w_2 \derivesby{q.12} w_1 \dollari{4}_q \hat{Y}\bar{Y} \dollari{3}_q \hat{Z}
\dollari{5}_q w_2 \\ \derivesby{q.13} w_1 \dollari{4}_q \hat{Y}\bar{Y}
\dollari{3}_q \hat{Z}\bar{Z} \dollari{5}_q w_2 \derivesby{q.14} w_1
\dollari{4}_q \hat{Y}\bar{Y} \hat{Z}\bar{Z} \dollari{5}_q w_2 \derivesby{q.15}
w_1 \hat{Y}\bar{Y} \hat{Z}\bar{Z} \dollari{5}_q w_2  \derivesby{q.16} w_1
\hat{Y}\bar{Y} \hat{Z}\bar{Z} w_2.
\end{multline*}
Here, again, $w_1, w_2\in V_\Gamma^*$.  We remark that if $w_1 \hat{X}\bar{X}
w_2$ satisfies the normalization condition, all the strings shown in this
derivation sequence satisfy the normalization condition as well.

We claim that any derivations aside from the one shown above result in
unproductive branches.  In the next paragraphs we will briefly expose the
rationale behind this statement.

It should be clear by now that the principal way to control the loci of
insertions and deletions is to carefully construct the permitting and
forbidding contexts of the rules which are to be applied next. Consider the
contexts of $q.2$, which compel the rule $q.1$ to insert $\dollari{1}_q$
exactly before an instance of $\hat{X}$.  Similarly, the application of $q.2$
will produce a valid string only $\dollari{2}_q$ is inserted after an instance
of $\bar{X}$.  Note that at the moment we have not assured yet that
$\dollari{1}_q$ and $\dollari{2}_q$ are located before and after the same pair
$\hat{X}\bar{X}$.

Once the symbols $\dollari{1}_q$ and $\dollari{2}_q$ have been inserted into
the string before and after some pairs $\hat{X}\bar{X}$, the rule $q.3$ will
erase an instance of $\hat{X}$; the normalization condition will assure that
the erased symbol was located exactly after $\dollari{1}_q$.

The contexts of the rule $q.4$ guarantee that the proper instance of $\hat{X}$
has already been erased (the permitting condition $\dollari{1}_q\bar{X}$) and
will block the rule in the cases when $\dollari{1}_q$ is inserted between
$\hat{X}$ and $\bar{X}$ and when $\dollari{2}_q$ is inserted after
$\dollari{1}_q$; in the latter case $q.3$ would be blocked.

The rule $q.5$ assures the proper localization of $\dollari{3}_q$, i.e., before
$\dollari{1}_q$.  The effect of rules $q.4$ and $q.5$ is thus to substitute
$\dollari{1}_q$ with $\dollari{3}_q$.

The rule $q.6$ finalizes the erasure of $\hat{X}\bar{X}$.  $\bar{X}$ cannot be
erased before $\dollari{1}_q$ is erased, which also guarantees in the long run
that $\dollari{3}_q$ has been properly positioned.

In the good case we would now expect to have $\dollari{3}_q \dollari{2}_q$;
however, we have not yet assured that the symbol $\dollari{2}_q$ was inserted
after the same pair $\hat{X}\bar{X}$ as the pair before which $\dollari{3}_q$
was inserted.  This is the job of the rule $q.7$: the symbol $\dollari{4}_q$ is
only inserted when there actually is a pair $\dollari{3}_q\dollari{2}_q$ in the
string. If this is not the case, $\dollari{4}_q$ will not be inserted; it is
easy to verify that the subsequent rules will be rendered inapplicable in this
case, thus blocking the computation at a string with non-terminals.

The rule $q.8$ assures proper positioning of $\dollari{4}_q$.  Again, if
$\dollari{2}_q$ is not erased here, all further rule applications will be
blocked.

The rule $q.9$ inserts yet another service symbol $\dollari{5}_q$. The contexts
of this rule assure that it is only applied at the proper time.  The contexts
of $q.10$ assure the proper positioning of $\dollari{5}_q$.  An application of
the rule $q.10$ inserts an instance of $\hat{Y}$ and, by the normalization
condition, this will not block any computations only if $\hat{Y}$ is inserted
to the left or to the right of a service symbol of the $\$$ family.  The
contexts of the rules $q.10$ and $q.11$ assure that $\hat{Y}$ is placed exactly
between $\dollari{4}_q$ and $\dollari{3}_q$.  Similarly, the contexts of $q.11$
and $q.12$ guarantee the proper positioning of $\bar{Y}$. And again, the
situation is exactly the same for the rules $q.12$ and $q.13$, which are
guaranteed to place the corresponding symbols at the proper sites.

A very important remark is that, once we have more than one non-special symbol
between two special symbols (i.e, the symbols of the $\$$ family), the
normalization condition becomes relevant for the substring between the special
symbols.  This limits the number of times the insertion rules $q.10$ to $q.13$
can be applied.

The contexts of the rule $q.14$ provide the final verification of the
configuration of the string around the symbols $\dollari{4}_q$,
$\dollari{3}_q$, and $\dollari{5}_q$.  Once $\dollari{3}_q$ is erased, $q.15$
erases $\dollari{4}_q$.  Note how the contexts of $q.15$ only allow this rule
to be applied after the $\dollari{3}_q$ has been erased.  Finally, $q.16$
finalizes the clean-up of the string.

Now in order to finish the proof of the theorem we show the inclusion
$SC_{2,2}INS^{0,0}_1DEL^{0,0}_1\subseteq RE$. Since
$INS^{0,0}_1DEL^{0,0}_1\subseteq CF$ (see~\cite{SV2-2}) we obtain the inclusion
$SC_{2,2}INS^{0,0}_1DEL^{0,0}_1\subseteq SC_{2,2}\subseteq RE$.
\end{proof}

\section{Random Context Insertion-Deletion Systems}

In this section we consider random context insertion-deletion systems. We show
that such systems having rules of size $(2,0,0;1,1,0)$ are computationally
complete, while systems of size $(1,1,0;p,1,1)$, $p>0$, are not.

\subsection{Computational Completeness}

We start by showing the computational completeness of the family
$RC\,INS^{0,0}_{2}DEL^{1,0}_{1}$.

\begin{thm}\label{thm:rc200110}
$RC\,INS^{0,0}_{2}DEL^{1,0}_{1}=RE$.
\end{thm}

\begin{proof}

We show only the inclusion $RC\,INS^{0,0}_{2}DEL^{1,0}_{1}\supseteq RE$. For
the converse inclusion we invoke the Church-Turing thesis. We will simulate an
arbitrary grammar in the special Geffert normal form. Let $G=(V,T,S,R)$ be such
a grammar. We will simulate $G$ with the following random context
insertion-deletion system $\Gamma=(V_\Gamma,T,\{S\},R_\Gamma)$.

The alphabet $V_\Gamma$ is constructed in the following way: $V_\Gamma=V\cup
N_\Gamma.$
The set $N_\Gamma$ contains the non-terminal service symbols we add to
$\Gamma$.  These symbols are the hatted and barred version of $A$,
$B$, $C$, and $D$, as well as the additional symbols of $\$_r$,
$\#_q$, and $f_r$ families, added on per-rule basis, as shown in the
following paragraphs.

The set of rules $R_\Gamma$ is constructed in the following way:
\begin{itemize}
\item for every rule $(p:X\ra cY)\in R$, $c\in T$ we add the following
  to $R_\Gamma$:
  \[
  \begin{array}{rrll}
    \left\{\right.&p.1:&\RCrule{\fins{cY}}{\{X\}}{(N'\cup N_\Gamma)\backslash\{X\})},&\\
    &p.2:&\RCrule{\lcdel{Y}{X}}{\emptyset}{\emptyset}&\left.\right\},
  \end{array}
  \]
\item for every rule $(p:X\ra Z Y)\in R$, $Z\in N$ we
  consider the rule $p':(X\ra \hat{Z}\bar{Z} Y)$; this rule can be
  rewritten as $\{X\ra \hat{Z} X', X'\ra \bar{Z} Y\}$; we then apply
  the same reasoning as in the previous paragraph,

\item for every rule $(q:X\ra Yc)\in R$, $c\in T$ we add the following
  to $R_\Gamma$:
  \[
  \begin{array}{rrll}
    \left\{\right.&q.1:&\RCrule{\fins{\#_q\#'_q}}{\{X\}}{(N'\cup N_\Gamma)\backslash\{X\}},&\\
    &q.2:&\RCrule{\lcdel{\#'_q}{X}}{\emptyset}{\emptyset},&\\
    &q.3:&\RCrule{\fins{Yc}}{\left\{\#_q\right\}}{\{X,Y\}},&\\
    &q.4:&\RCrule{\lcdel{c}{\#'_q}}{\emptyset}{\emptyset},&\\
    &q.5:&\RCrule{\fdel{\#_q}}{\emptyset}{\{\#'_q\}}&\left.\right\},\\
  \end{array}
  \]
\item for every rule $(q:X\ra Y Z)\in R$, $Z\in N$ we
  consider the rule $q':(X\ra Y \hat{Z}\bar{Z})$; this rule can be
  rewritten as $\{X\ra X' \bar{Z}, X'\ra Y \hat{Z}\}$; we then apply
  the same reasoning as in the previous paragraph,

\item for every rule $(r:UV\ra \lambda)\in R$, $(U,
  V)\in\{(A,B),(C,D)\}$ we add the following to $R_\Gamma$:
  \[
  \begin{array}{rrll}
    \left\{\right.&r.1:&\RCrule{\fins{\$_r \$'_r}}{\emptyset}{N_\Gamma\backslash\{\hat{U},\bar{U},\hat{V},\bar{V}},&\\
    &r.2:&\RCrule{\lcdel{\hat{U}}{\$_r}}{\emptyset}{\emptyset},&\\
    &r.3:&\RCrule{\fins{\$''_r}}{\{\$'_r\}}{\{\$_r,\$''_r,\$'''_r\}},&\\
    &r.4:&\RCrule{\lcdel{\$''_r}{\$'_r}}{\emptyset}{\emptyset},&\\
    &r.5:&\RCrule{\lcdel{\$''_r}{\bar{U}}}{\emptyset}{\{\$'_r\}},&\\
    &r.6:&\RCrule{\fins{\$'''_r}}{\{\$''_r\}}{\{\$'_r,\$'''_r\}},&\\
    &r.7:&\RCrule{\lcdel{\$'''_r}{\$''_r}}{\emptyset}{\emptyset},&\\
    &r.8:&\RCrule{\lcdel{\$'''_r}{\hat{V}}}{\emptyset}{\{\$''_r\}},&\\
    &r.9:&\RCrule{\lcdel{\$'''_r}{\hat{A}}}{\emptyset}{\{\$''_r\}},&\\
    &r.10:&\RCrule{\lcdel{\$'''_r}{\hat{C}}}{\emptyset}{\{\$''_r\}},&\\

    &r.11:&\RCrule{\fins{f_r f'_r}}{\{\$'''_r\}}{\{\$''_r, f'_r\}},&\\
    &r.12:&\RCrule{\lcdel{\bar{V}}{f_r}}{\emptyset}{\emptyset},&\\
    &r.13:&\RCrule{\fins{\dollarfour_r}}{\{f'_r\}}{\{f_r, \dollarfour_r\}},&\\
    &r.14:&\RCrule{\lcdel{\dollarfour_r}{\$'''_r}}{\emptyset}{\emptyset},&\\
    &r.15:&\RCrule{\lcdel{\dollarfour_r}{\bar{V}}}{\emptyset}{\{\$'''_r\}},&\\
    &r.16:&\RCrule{\lcdel{\dollarfour_r}{f'_r}}{\emptyset}{\{\$'''_r\}},&\\
    &r.17:&\RCrule{\fdel{\dollarfour_r}}{\emptyset}{\{f'_r\}},&\\

    &r.18:&\RCrule{\lcdel{\emptyset}{\hat{U}}}{\emptyset}{V_\Gamma\backslash\left(T\cup\{\hat{A},\hat{C}\}\right)}&\left.\right\}.\\
  \end{array}
  \]
\end{itemize}

Consider the application of the rule $p:X\ra cY$ to $w_1 X w_2$,
$w_1,w_2\in (T\cup N)^*$.  Remember that, for any string $w$
generated by a grammar in the special Geffert normal form, the following
holds: $|w|_{Z}\leq 1, \forall Z\in N'$.

The correct sequence of $\Gamma$ events simulating the rule $p$ is as
follows:
\[
w_1 X w_2\derivesby{p.1} w_1 cYX w_2\derivesby{p.2} w_1 cY w_2.
\]

Note that $p.1$ will only be applicable when there is one copy
of $X$ in the string and where there is no other symbol from $N'\cup N_\Gamma$.
This, among others, prohibits starting the simulation of rule $p$ if
another rule is being simulated.

Now suppose that $cY$ is inserted in a different place in the string
(not before $X$).  In this case $p.2$ will not be applicable in the
new configuration and both $X$ and $Y$ will remain in the string.
Since a simulation of a rule cannot start when there is more than one
symbol from $N'$ in the string, the derivation halts at this invalid
string.

The simulation of the rule $q:X\ra Yc$ takes more steps.  The reason
is that we use deletion rules to assure proper location of insertion
sites, and deletion rules have left contexts, while in the string $Yc$
the symbol $c$ is to the right of $Y$.

The correct sequence of $\Gamma$ events is the following.
\[
\begin{array}{l}
w_1 X w_2 \derivesby{q.1} w_1 \#_q \#'_q X w_2 \derivesby{q.2} w_1
\#_q \#'_q w_2 \derivesby{q.3} w_1 \#_q Yc \#'_q w_2 \\\derivesby{q.4}
w_1 \#_q Yc w_2 \derivesby{q.5} w_1 Yc w_2.
\end{array}
\]

We will now focus on analysing the situations when the insertions do
not happen at desired loci.  Suppose that $q.1$ does not insert
$\#_q\#'_q$ before $X$.  In this case $q.2$ is not applied.  The rule
$q.3$ cannot be applied, because there is an instance of $X$ in the
string.  The rule $q.4$ cannot be applied, because $\#'_q$ is located
immediately after $\#_q$.  The rule $q.5$ is inapplicable, because
there is an instance of $\#'_q$ in the string.  Therefore, if
$\#_q\#'_q$ is not inserted right before $X$, all three symbols from
$N'$ (namely, $X$, $\#_q$, and $\#'_q$) remain in the string.  This
prevents any further simulation of rules and the derivation halts at
an invalid string.

The other insertion rule is $q.3$.  Suppose that $Yb$ is inserted in a
place other than between $\#_q$ and $\#'_q$.  In that case the rule
$q.4$ will never be applied and neither will $q.5$ (because the symbol
$q.4$ will not be deleted).  The string will contain three symbols
from $N'\cup N_\Gamma$ ($Y$, $\#_q$, and $\#'_q$), which makes any further
simulation impossible.  Again, the derivation halts at an invalid
string.

The simulation of the rule $UV\ra \lambda$ is the longest.  The
two problems with this rule are that we need to remove the symbols $U$
and $V$ only if they are located in a certain order and that we need
to remove exactly one instance of $U$ and one instance of $V$.

The correct simulation sequence is as follows:
\[
\begin{array}{l}
u_1 \hat{U} \bar{U} \bar{V} \hat{V} u_2\derivesby{r.1}
u_1 \hat{U} \$_r \$'_r \bar{U} \hat{V} \bar{V} u_2\derivesby{r.2}
u_1 \hat{U} \$'_r \bar{U} \hat{V} \bar{V} u_2\derivesby{r.3}
u_1 \hat{U} \$''_r \$'_r \bar{U} \hat{V} \bar{V} u_2\\\derivesby{r.4}
u_1 \hat{U} \$''_r \bar{U} \hat{V} \bar{V} u_2\derivesby{r.5}
u_1 \hat{U} \$''_r \hat{V} \bar{V} u_2\derivesby{r.6}
u_1 \hat{U} \$'''_r \$''_r \hat{V} \bar{V} u_2\derivesby{r.7}
u_1 \hat{U} \$'''_r \hat{V} \bar{V} u_2\\\derivesby{r.8}
u_1 \hat{U} \$'''_r \bar{V} u_2\derivesby{r.11}
u_1 \hat{U} \$'''_r \bar{V} f_r f'_r u_2\derivesby{r.12}
u_1 \hat{U} \$'''_r \bar{V} f'_r u_2\\\derivesby{r.13}
u_1 \hat{U} \dollarfour_r \$'''_r \bar{V} f'_r u_2\derivesby{r.14}
u_1 \hat{U} \dollarfour_r \bar{V} f'_r u_2 \derivesby{r.15}
u_1 \hat{U} \dollarfour_r f'_r u_2 \\\derivesby{r.16}
u_1 \hat{U} \dollarfour_r u_2 \derivesby{r.17}
u_1 \hat{U} u_2,
\end{array}
\]
where $u_1, u_2\in(T\cup\{\bar{U},\hat{U}\mid U\in\{A,B,C,D\}\})^*$.
Note that $r.18$ will eventually erase all instances of $\hat{U}$ when
the string will only contain symbols from $T\cup \{\hat{A},\hat{C}\}$.

The simulation of the erasing rule starts with fixing a $\$'_r$ at the
site between an instance of $\hat{U}$ and $\bar{U}$.  An instance of
$\$''_r$ is then inserted before $\$'_r$ and erases both $\$'_r$ and
the $\bar{U}$ which is located next to it.  Next $\$'''_r$ is
inserted, which erases both $\$''_r$ and the $\hat{V}$ which should
follow immediately in a valid string.  Then we fix a site immediately
after a $\bar{V}$ by inserting $f_r f'_r$ and having $\bar{V}$ erase
$f_r$.  This allows to eventually insert $\dollarfour_r$ whose mission
is to erase $\$'''_r$ and a $\bar{V}$.  We demand that $\dollarfour_r$
erase $f'_r$ in order to assure that, on the one hand, $\dollarfour_r$
has erased $\bar{V}$ and, on the other hand, that $f'_r$ is right
after the currently processed sequence $\hat{U}\bar{U}\hat{V}\bar{V}$.

Note that, when erasing a sequence $\hat{U}\bar{U}\hat{V}\bar{V}$, we
do not erase the $\hat{U}$.  Instead, we rely on rule $r.17$ to do
that when the string only contains terminals and, possibly, hatted
versions of $A$ and $C$.  It is not necessary to erase $\hat{U}$
because, by erasing the pairs $\bar{U}\hat{V}$ and discarding
$\hat{U}$ and $\bar{V}$, we already do the necessary checking of
whether the non-terminals $\{A,B,C,D\}$ are organised in proper pairs.
The only situation when the left-over instance of $\hat{A}$ or
$\hat{C}$ can interfere with other simulations of the rule
$r:UV\ra\lambda$ is when there are nested pairs, like
$\hat{A}\bar{A}\hat{C}\bar{C}\hat{D}\bar{D}\hat{B}\bar{B}$.  In this
case, the inner pair will be erased first, leaving a $\hat{C}$.  When
the $\$'''_r$ corresponding to the rule $AB\ra\lambda$ will need to
erase $\$''_r$ and $\hat{B}$ from the string, it can also remove the
left-over $\hat{C}$ by rule $r.10$.  Thus, the garbage symbols
$\hat{A}$ and $\hat{C}$ cannot disrupt the simulation of other erasing
rules.

On the other hand, observe that $\$'''_r$ is programmed to erase
either $\$''_r \hat{V}$, or $\$''_r \hat{A} \hat{V}$, or $\$''_r
\hat{C} \hat{V}$.  Two hatted symbols one after another cannot be
produced after the simulation of the context-free rules of the
grammar, because those rules only produce pairs of hatted and barred
symbols. This means that $\$'''_r$ will successfully handle the
garbage symbols resulting from the formerly existent nested pair, but
will not be able to produce unwanted effects in other parts of the
string.

We will explicitly remark that, according to the prohibiting set of
$r.1$, the simulation of a rule $UV\ra\lambda$ can only start once the
simulation of all other types of rules have completed.

We will now try to see whether it is possible for the described
simulation process to produce wrong results.  Suppose that $r.1$ does
not insert the pair $\$_r\$'_r$ after $\hat{U}$.  In this case $r.2$
will not be applicable.  Furthermore, the rule $r.3$ will not be
applicable either, because the symbol $\$_r$ is still in the string.
Obviously, the rules $r.4$, $r.5$, and $r.6$ cannot be considered
because there is no symbol $\$''_r$.  Since the symbol $\$'''_r$ is
not inserted either, the rules $r.7$ to $r.11$ are inapplicable.
Because the pair $f_r f'_r$ is not inserted, the rules $r.12$ and
$r.13$ cannot be applied.  Finally, because no rule inserts
$\dollarfour_r$, the rules $r.14$ to $r.17$ are not applicable either.
Nevertheless, the symbols $\$_r$ and $\$'_r$ are still in the string,
which makes any other simulations impossible, which renders the whole
branch of computation invalid.

Similar considerations apply to the insertion rule $r.3$.  In the case
$\$''_r$ does not arrive at the proper site, $\$'_r$ is not removed,
which blocks further applications of rules.  In an analogous way, the
insertion of $\$'''_r$ at an improper site renders the simulation
invalid.  And for the same reasons, the insertion of $\dollarfour_r$
anywhere but immediately before $\$'''_r$ blocks the computation.
Note that, before allowing a the symbol $\$^{(i)}_r$ to erase the
corresponding hatted or barred version of $U$ or $V$, we demand that
it erase the symbol $\$^{(i-1)}_r$.  In this way we guarantee that the
symbols of the $\$_r$ family are always inserted inside the same
sequence of $\hat{U}\bar{U}\hat{V}\bar{V}$ as the one that was fixed
by the initial insertion of $\$_r\$'_r$.

The only insertion rule left is $r.11$.  The rule $r.12$ and the
forbidding contexts of the following rules guarantee that the symbol
$f'_r$ is located after an instance of $\bar{V}$.  There is however no
guarantee that this will be the instance of $\bar{V}$ which belongs to
the processed sequence $\hat{U}\bar{U}\hat{V}\bar{V}$.  Suppose that
$f'_r$ is indeed inserted after a different instance of $\bar{V}$.  In
this case the rule $r.16$ will be inapplicable, because there can be
only one instance of $f'_r$ in the string, and this instance is
positioned improperly.  This will not allow $\dollarfour_r$ to
eventually be erased by the rule $r.17$ and will therefore block any
further computations.


\end{proof}

\subsection{Computational Incompleteness Results}
In this section we will show that  random context insertion-deletion systems of
size $(1,1,0;p,1,1)$ are not capable of generating all recursively enumerable
languages, hence they are computationally incomplete.

The proof below is based on the observation that it is impossible to control
the number of applications of an insertion rule $r:\RCrule{\lcins{x}{y}}{P}{Q}$
when $y$ is already present in the string.

Obviously, in the case $P=Q=\emptyset$, if the rule can be applied once, it can
also be applied any number of times. We also remark that including $x$ in
either $P$ or $Q$ gives no advantage; including $x$ in the permitting context
is redundant, while including it in the forbidding context would just make the
rule never applicable. Hence, there are two remaining cases: $y\not \in Q$, and
$y\in Q$.  When $y\not\in Q$, the situation is exactly the same as with
$P=Q=\emptyset$ from the point of view of controlling the number of
applications.  On the other hand, when $y\in Q$, the rule $r$ will only be
applicable when the string contains no $y$.  So, when symbol $y$ is already
present in the string, rule $r$ can be applied any number of times.

\begin{thm}\label{thm:rc110p11}
$RE\backslash RC\,INS^{1,0}_1DEL^{1,1}_p\neq\emptyset$.
\end{thm}

\begin{proof}
We show that the language $L=(ab)^+$ cannot be generated with such systems. We
shall do the proof by contradiction. Suppose that there is a random context
insertion deletion system $\Gamma = (V, T, S, R)$, $\Gamma\in
RC\,INS^{1,0}_1DEL^{1,1}_p$, such that $L(\gamma)=L$. Without restricting
generality, we will not allow $R$ to contain rules erasing terminal symbols,
see~\cite{AKRV10b} or~\cite{VerlanH} for more details.

Consider a derivation $S\derivesby{+} w$, where $w\in (ab)^+$.  We will focus
on a pair of symbols $ab$: $w = w_1 ab w_2$, where $w_1, w_2\in (ab)^*$.

Since $R$ does not contain rules erasing terminals, the instance of
$b$ which belongs to the current pair was either inserted at this site
by an insertion rule or was a part of the axiom $S$.  We will not
examine this latter case, and will instead suppose that the word $w$
is long enough.  The rule which inserted $b$ is of the form
$p:\RCrule{\lcins{x}{b}}{P}{Q}$.  It is very important to remark at
this point that the permitting and forbidding contexts of the rule $p$
cannot influence the position at which the insertion takes place.  We
are therefore not interested in the case $x=\lambda$, because such a
rule would be allowed to insert the $b$ at any site on the string,
and, in particular, we could obtain the string $w'=w_1 ba w_2\not \in
(ab)^+$.

So, we have the following derivation:
\[
S\derivesby{*} w'_1 x w'_2 \derivesby{p} w'_1 x b w'_2 \derivesby{*} w_1 ab w_2 = w.
\]
where $w'_1, w'_2\in V^*$. Since $w\in (ab)^+$ we can suppose that $|w_1'|_b>0$
(as well as $|w_2'|_b>0$). In this case, according to our remark above, rule
$p$ can be applied an arbitrary number of times yielding the following
derivation:
$$
 S\derivesby{*} w'_1 x w'_2 \derivesby{p} w'_1 x b w'_2\derivesby{p^+} w'_1 x b^+ w'_2\derivesby{*} w_1 ab^+
w_2.
$$

We observe that the last part of the derivation is possible because a deletion
rule of size $(p,1,1)$ cannot distinguish if there is one or more $b$ present
in the left or right context.

We therefore conclude that, if $S\derivesby{*} w_1 ab w_2$, then
$S\derivesby{*} w_1 ab^+ w_2$ and $L(\Gamma)\neq (ab)^+$, which contradicts our
initial supposition.
\end{proof}

\section{Conclusion}
In this paper we introduced the mechanism of conditional application of rules
for insertion-deletion systems based on semi-conditional and random context
conditions. This mechanism permitted to achieve an increase in the
computational power, especially for the case of semi-conditional
insertion-deletion systems, which are showed to characterize the family of
recursively enumerable languages with rather simple rules: context-free
insertion or deletion of one symbol. This shows that the control mechanism is
in some sense more powerful than the graph control, because in the latter case
only the recursively enumerable sets of numbers can be generated~\cite{AKRV11}.
The form of the rules and of the control mechanism permit to consider it as a
particular case of Networks of Evolutionary Processors~\cite{NEP}, so the
obtained result can be transcribed in that area as well.

In the case of random-context conditions we obtained results exhibiting an
interesting asymmetry between the computational power of insertion and deletion
rules: systems having rules of size $(2,0,0;1,1,0)$ are computationally
complete, while those having the size $(1,1,0;2,0,0)$ (and more generally of
size $(1,1,0;p,1,1)$) are not. This result is surprising because until now all
the obtained results in the area of insertion-deletion systems with or without
additional controls were analogous if parameters of insertion and deletion
rules were interchanged. The characterization of these classes gives an
interesting topic for a further research.

We remark that all the constructions were additionally tested using a
self-developed dedicated simulator that is available ``on request'' from the
corresponding author. These tests permitted to verify that for a single rule
application, all possible evolutions except the one corresponding to the
correct one introduce in the string (groups of) non-terminal symbols that
cannot be removed anymore.

\end{document}